%% file: WirelessSINR-Leader.tex
\documentclass[11pt,twocolumn]{article}
\usepackage{times,mathptmx}
\usepackage{fullpage}
\usepackage{amsthm,amsmath}
\usepackage{amssymb}

\usepackage[ruled]{algorithm}
\usepackage[noend]{algpseudocode}

\usepackage{epsfig}

\newtheorem{definition}{Definition}
\newtheorem{theorem}{Theorem}
\newtheorem{corollary}{Corollary}
\newtheorem{lemma}{Lemma}

\newtheorem{proposition}{Proposition}

\newtheoremstyle{redstyle}
     {3pt}
     {3pt}
     {\color{black}}
     {}
     {\color{red}\bfseries}
     {:}
     {.5em}
     {}
\theoremstyle{redstyle}

\newcommand{\m}{\mathcal}
\newcommand{\cT}{{\mathcal T}}
\newcommand{\cN}{{\mathcal N}}

\newcommand{\ep}{\varepsilon}
\newcommand{\eps}{\varepsilon}

\newcommand{\pred}{\text{pred}}

\newcommand{\leader}{\text{ld}}
\newcommand{\DIR}{\text{DIR}}
\newcommand{\sselector}{SINR-selector}

\newcommand{\tju}[1]{} 
\newcommand{\labell}[1]{\label{#1}} 
\renewcommand{\paragraph}[1]{\vspace*{0.5ex}\noindent {\bf #1}}
\newcommand{\remove}[1]{}

\newcommand{\polylog}{\text{\! polylog \!}}
\newcommand{\dist}{\text{dist}}

\newcommand{\NAT}{{\mathbb N}}
\newcommand{\INT}{{\mathbb Z}}

\usepackage{color}

\newcommand{\tj}[1]{#1}
\newcommand{\comment}[1]{}

\newcommand{\LocalLeader}{\textsc{Local Leader Election}}
\newcommand{\LocalLearning}{\textsc{Local Learning}}
\newcommand{\NeighborL}{\textsc{Neighborhood Learning}}
\newcommand{\InterC}{\textsc{Inter-Box Communication}}
\newcommand{\state}{\text{st}}
\newcommand{\boxx}{\text{box}}

\newcommand{\forward}{\text{forward}}
\newcommand{\waitback}{\text{wait-back}}
\newcommand{\backward}{\text{back}}
\newcommand{\waitconf}{\text{wait-conf}}
\newcommand{\confirm}{\text{confirm}}
\newcommand{\stopp}{\text{stop}}
\newcommand{\suc}{\text{succ}}
\newcommand{\GlobalLeader}{GlobalLeader}
\newcommand{\MultiAlg}{Mutli-broadcast}

\begin{document}
\title{\bf Distributed backbone structure for deterministic algorithms in the
SINR model of wireless networks}
\author{    Tomasz Jurdzinski\footnotemark[2] \footnotemark[3]
    \and
    Dariusz R.~Kowalski\footnotemark[3] \footnotemark[4]
}
\date{}
\maketitle

\footnotetext[2]{Institute of Computer Science, U. Wroc{\l}aw, Poland.}

\footnotetext[3]{Department of Computer Science, U. Liverpool, UK.}

\footnotetext[4]{Institute IMDEA Networks, 28918 Madrid, Spain.}

\begin{abstract}
The Signal-to-Interference-and-Noise-Ratio (SINR) physical model is one of the legitimate models of wireless networks. Despite of the vast amount of study done in design and analysis of centralized algorithms supporting wireless communication under the SINR physical model,
little is known about distributed algorithms in this model, especially deterministic ones. In this work we construct, in a deterministic distributed way, a backbone structure on the top of a given wireless network, which can be used for transforming many algorithms designed in a simpler model of ad hoc broadcast networks without interference into the SINR physical model with uniform power of stations, without increasing their asymptotic time complexity.
The time cost of the backbone data structure construction  is only
$O(\Delta \polylog n)$ rounds, where $\Delta$ is roughly the inverse of network density and $n$ is the number of nodes in the whole network.
The core of the construction is a novel combinatorial structure called SINR-selector, which is introduced and constructed in this paper.
We demonstrate the power of the backbone data structure by using it for obtaining efficient
$O(D+\Delta \polylog n)$-round and $O(D+k+\Delta \polylog n)$-round deterministic distributed solutions
for leader election and multi-broadcast, respectively, where $D$ is the network diameter and $k$ is the number of messages to be disseminated.
\end{abstract}

\newcommand{\T}{\vspace*{-1ex}}

\T\T
\section{Introduction}

\T
In this work we study a fundamental problem how to transform algorithms designed and analyzed for
ad hoc networks without any interference (e.g., message passing or multicast networks)
to ad-hoc wireless networks under the
Signal-to-Interference-and-Noise-Ratio model (SINR).
A wireless network considered in this work consists of $n$ stations, also called nodes, with
uniform transmission powers, deployed in the two-dimensional
Euclidean space.
%
Stations act in synchronous rounds; in every communication round a station can either transmit
a message or listen to the wireless medium.
A communication (or reachability) graph of the network is the graph defined on network nodes
and containing links $(v,w)$ such that if $v$ is the only transmitter in the network then $w$ receives
the message transmitted by $v$ under the SINR physical model.
Each station initially knows only its own unique ID
in the range $\{1,\ldots,N\}$,
(estimated) location, and
parameters $N$ and $\Delta$, where $\Delta$ is the upper bound on the node degree
in the communication graph of the network, and corresponds roughly to the inverse
of the lower bound on network density.

We consider global communication tasks in the SINR wireless setting, and our objective is
to minimize time complexity (i.e., the number of rounds) of deterministic distributed solutions.
In order to overcome the impact of signal interference, we show how to compute a backbone data structure
in a distributed deterministic way, which also clusters nodes and implements efficient inter- and intra-cluster
communication. It allows to transform a variety of algorithms designed and analyzed in models with no-interference
to the SINR wireless model with only additive $O(\Delta \polylog n)$ overhead on time complexity.
We demonstrate such efficient transformation on two tasks: leader election and multi-broadcast
with small messages.


\T\T
\subsection{Previous and Related Results}

\T
\paragraph{SINR model.}
The Signal-to-Interference-and-Noise-Ratio
(SINR) physical model is currently the most popular framework deploying
physical wireless interference for the purpose of design and theoretical
analysis of wireless communication tasks.
There is a vast amount of work on centralized algorithms under the SINR model.
The most studied problems include connectivity, capacity maximization,
link scheduling types of problems (e.g.,\ \cite{FanghanelKRV09,Kesselheim11,AvinLPP09}).
See also the survey~\cite{WatSurv} for recent advances and references.

Recently, there is a growing interest in developing solutions to {\em local communication} problems,
in which stations are only required to exchange information with a subset of their neighbors.
Examples of such problems include local broadcast or local leader election.
A deterministic {\em local} broadcasting, in which nodes have to inform
only their neighbors in the corresponding reachability graph,
was studied in \cite{YuWHL11}.
The considered setting allowed power control by algorithms,
in which, in order to avoid collisions,
stations could transmit with any power smaller than the
maximal one.
Randomized solutions for contention resolution~\cite{KV10}
and local broadcasting~\cite{GoussevskaiaMW08} were also obtained.
%
%
Single hop topologies were also studied recently
under the SINR model, c.f.,~\cite{RichaSSZ}.

Randomized distributed solutions to local communication problems, as cited above,
can often be used as a basic tool for obtaining {\em randomized solutions} to {\em global communication} tasks,
i.e., tasks requiring information exchange throughout the whole network.
Examples of such task include multi-broadcast or leader election.
These problems are also related, though often not equivalent, to several graph-related problems
of finding a maximal/maximum independent set, minimal/minimum (connected) dominating set
(our backbone structure is an example of the latter, with additional useful properties).
Recently, an efficient randomized distributed solution was obtained to the problem
of finding
a constant-density dominating set
by Scheideler et al.~\cite{ScheidelerRS08}, however the model of that paper is slightly different than the one in this work,
i.e., it combines SINR with radio model.

Surprisingly, to the best of our knowledge, there is no published throughout study on
{\em deterministic distributed} solutions to any fundamental {\em global communication} problem
under the SINR model.\footnote{%
Very recently, we designed algorithms for a different setting and problem,
in which a single source must wake up the network by a broadcast message, c.f.,~\cite{JK12broadcast}
(unpublished manuscript, submitted).
That problem is very different in nature from the ones considered in this work,
and appeared more costly for $D,\Delta=o(n)$ (unless an additional knowledge
about locations of neighbors is provided to every node).
%
}
Such algorithms are especially important from perspective of maintaining network infrastructure
and applications to distributed systems.
Deterministic solutions are often desired in such cases, due to their reliability.
One could apply a naive round-robin algorithm to build a collision-free
solution to such problems,
but apparently it is extremely inefficient.

\paragraph{Radio network model.}
%
In the related {\em radio model} of wireless networks,
a message is successfully heard if there are no other simultaneous transmissions
from the {\em neighbors} of the receiver in the communication graph.
This model does not take into account the real strength of the received signals, and also the signals
from the outside of some close proximity.
In the geometric ad hoc setting, Dessmark and Pelc~\cite{DessmarkP07} were the first who studied
global communication problems, mainly broadcasting.
They analyzed the impact of local knowledge, defined as a range within which
stations can discover the nearby stations.
%
%
%
Emek et al.~\cite{EmekGKPPS09} designed a broadcast algorithm
working in time $O(Dg)$
in {\em unit-disc graph (UDG)} radio networks with eccentricity $D$ and granularity $g$, where
eccentricity was defined as the minimum number of hops to propagate the broadcast message throughout
the whole network and
granularity was defined as the
inverse of the minimum distance between any two stations.
Later, Emek et al.~\cite{EmekKP08} developed a matching lower bound $\Omega(Dg)$.
%
%
%
%
Leader election problem for geometric radio networks was studied e.g., by
Chung et al.\ \cite{ChungRW11} in the case of mobile devices.

Communication problems are well-studied in the setting of {\em graph radio model}, in which stations
are not necessarily deployed in a metric space.
%
Due to limited space and the fact that this research area is not directly related to the core of this work,
we refer the reader to the recent
literature on
deterministic~\cite{Censor-HillelGKLN11, CzumajRytter-FOCS-03, DeMarco-SICOMP-10, GalcikGL09, KP-DC-07, KowalskiP09}.
and randomized~\cite{Bar-YehudaGI92,CzumajRytter-FOCS-03, KP-DC-07,KushilevitzM98} solutions.

\T\T
\subsection{Our Results}



\T
The main result of this work is a deterministic and distributed algorithm
constructing a complex backbone distributed data structure, consisting of a combinatorial
structure with local operations on it, that is aimed to support wireless global communication tasks;
for brief description see Section~\ref{s:model}, and for detail implementation and
properties we refer the reader to Section~\ref{s:backbone}.
The construction is in $O(\Delta \polylog n)$ rounds.
The algorithm constructing the backbone network use a novel concept of SINR-selectors,
which are specific efficient schedules for ad hoc one-hop communication.
We define and construct them in Section~\ref{s:selector}.

Our work can be also viewed as a deterministic distributed implementation of a MAC layer,
introduced by Kuhn et al.~\cite{KuhnLN11}, under the SINR model.
It also allows to transform several algorithms designed for networks without interference
to the SINR wireless model, with an additive overheads $O(\Delta \polylog n)$ (coming from
spanning the backbone data structure),
and the number of required parallel intra-cluster convergecast communication tasks multiplied by $O(\Delta)$.

In many cases the number of such parallel local convergecasts can be lowered to $\polylog n$
or even a constant. As examples, we show that this can be done for the problems of leader election
and multi-broadcast (with small messages), and result in almost optimal
deterministic distributed $O(D+\Delta \polylog n)$-round and $O(k+D+\Delta \polylog n)$-round algorithms,
respectively, c.f., Section~\ref{s:alg}.
They are shown to be optimal in the SINR model up to a polylogarithmic factor,
i.e., they require $\Omega(D+\Delta)$ and $\Omega(k+D+\Delta)$, respectively, c.f., Section~\ref{s:lower}.
In particular, for the purpose of multi-broadcast protocol we can transform a property of greedy geometric
routing, proved for networks without collisions in~\cite{CidonKMP95}, to the SINR
model.


\vspace*{-3ex}
\section{Model and Notation}
\label{s:model}

\T
Throughout the paper, $\NAT$ denotes the set of natural numbers and $\INT$
denotes the set of integers.
For $i,j\in\NAT$, we use the notation $[i,j]=\{k\in\NAT\,|\,i\leq k\leq j\}$
and $[i]=[1,i]$.
%

We consider a wireless network consisting of $n$ {\em stations}, also called {\em nodes},
deployed into a two dimensional
Euclidean space and communicating by a wireless medium.
%
%
All stations have unique integer IDs in set $[N]$,
where $N$ is an integer model parameter. 
Stations
are denoted by letters $u, v, w$, which simultaneously
denote their IDs.
%
%
%
Stations are located on the plane with {\em Euclidean metric} $\dist(\cdot,\cdot)$,
and each station knows its coordinates.
%
%
Each station $v$ has its {\em transmission power} $P_v$, which is a positive real number.
There are three fixed model parameters: path loss
$\alpha> 2$,
threshold $\beta\ge 1$, and ambient noise $\cN\ge 1$.
The $SINR(v,u,\cT)$ ratio, for given stations $u,v$ and a set of (transmitting) stations $\cT$,
is defined as follows:
\vspace*{-2ex}
\begin{equation}\label{e:sinr}
SINR(v,u,\cT)
=
\frac{\frac{P_v}{\dist(v,u)^{\alpha}}}{\cN+\sum_{w\in\cT\setminus\{v\}}\frac{P_w}{\dist(w,u)^{\alpha}}}
\end{equation}

\vspace*{-1.5ex}
In the {\em Signal-to-Interference-and-Noise-Ratio model} (SINR) considered in this work,
station $u$ {\em successfully receives}, or {\em hears}, a message from station $v$ in a round if
$v\in \cT$, $u\notin \cT$, and:
\begin{itemize}
\vspace*{-1.3ex}
\item
$SINR(v,u,\cT)\ge\beta$, where $\cT$ is the set of stations transmitting at that time, and
\vspace*{-1.3ex}
\item
$P_v\dist^{-\alpha}(v,u)\geq (1+\eps)\beta\cN$,
\end{itemize}
\vspace*{-1.3ex}
where $\eps>0$ is a fixed {\em sensitivity parameter} of the model.
The above definition is common in the literature, c.f.,~\cite{KV10}.\footnote{%
The first condition is a straightforward application of the SINR ratio,
comparing strength of one of the received signals with the remainder.
The second condition enforces the signal to be sufficiently strong in order to be
distinguished from the background noise, and thus to be decoded.
}


In the paper, we make the following assumptions for the sake of
clarity of presentation: $\beta=1$,
$\cN=1$.
In general, these assumptions can be dropped without harming the asymptotic performances of
the presented algorithms and lower bounds formulas.
(Each time when these values may have an impact on the actual asymptotic complexity of algorithms
or lower bounds, we will discuss it in the paper.)

\paragraph{Ranges and uniformity.}
The {\em communication range} $r_v$ of a station $v$ is the radius of the circle in which a message transmitted
by the station is heard, provided no other station transmits at the same time.
A network
is
{\em uniform} when ranges (and thus transmission powers) of all stations are equal,
or {\em nonuniform} otherwise.
In this paper, only uniform networks are considered.
For clarity of presentation
we make an assumption that all powers are equal to $1$, i.e., $P_v=1$ for each $v$.
The assumed value $1$ of $P_v$ can be replaced by any fixed positive value without changing
asymptotic formulas for presented algorithms and lower bounds.
Under these assumptions, $r_v=r=(1+\eps)^{-1/\alpha}$
for each station $v$.
%
%
%
%
The {\em range area} of a station with range $r$
located at point $(x,y)$ is defined as a ball of radius $r$ centered at $(x,y)$.

\noindent
{\bf Communication graph and graph notation.}
The {\em communication graph} $G(V,E)$, also called the {\em reachability graph}, of a given network
consists of all network nodes and edges $(v,u)$ such that $u$ is in the range area of $v$.
%
Note that the communication graph is symmetric for uniform networks, which are considered
in this paper.
By a {\em neighborhood} $\Gamma(u)$ of a node $u$ we mean the set
of all 
neighbors of $u$, i.e., the set $\{w\,|\, (w,u)\in E\}$ in the communication graph $G(V,E)$
of the underlying network.
The {\em graph distance} from $v$ to $w$ is equal to the length of a shortest path from $v$ to $w$
in the communication graph (where the length of a path is equal to the number of its edges).
The {\em eccentricity} of a node
is the maximum graph
distance from this node to all other nodes
(note that the eccentricity is of order of the diameter if the communication
graph is symmetric --- this is also the case in this work).
%

We say that a station $v$ transmits {\em $c$-successfully} in a round
$t$ if $v$ transmits a message in round $t$ and this message is heard by
each station $u$ in a distance smaller or equal to $c$ from $v$. We say that a station
$v$ transmits {\em successfully} in round $t$ if it transmits $r$-successfully, i.e.,
each of its neighbors in the communication graph can successfully receive its message.
Finally, $v$ transmits {\em successfully} to $u$ in round $t$ if $v$ transmits
a message in round $t$ and $u$ successfully receives this message.

\paragraph{Synchronization.}
It is assumed that algorithms work synchronously in rounds, each station can
either act as a sender or as a receiver in a round.
All stations are active in the beginning of computation.
%

\paragraph{Collision detection and channel sensing.}
We consider the model without {\em collision detection}, that is,
if a station $u$ does not receive a message in a round $t$, it gets no information from the
physical wireless layer
whether any other station was transmitting
in that round
and about the value of $SINR(v,u,\m{T})$, for any station $u$, where $\m{T}$ is the set
of transmitting stations in round $t$.


\noindent
{\bf Communication problems and complexity parameters.}
We consider two global communication problems:
leader election and multi-broadcast.
{\em Leader election} is defined in the following way: initially all
stations of a network have
the same status non-leader and the goal is for all nodes but one to keep this status
and for the remaining single node to get the status leader. Moreover, all nodes must
learn the ID of the leader.

The {\em multi-broadcast} problem is to disseminate $k$ distinct 
messages initially stored at arbitrary nodes,
to the entire network (it is allowed that more than one message is stored
in a node). 

As a generic tool for these and possibly other global communication tasks,
we design a deterministic distributed algorithm building a backbone data structure, which
is a connected dominating set (i.e., a connected subgraph such that each node
is either in the subgraph or is a neighbor of a node in the subgraph, also called a CDS)
with constant degree, constant approximation of
the size of a smallest CDS,
and diameter proportional to the diameter of the whole communication graph.
Additionally, our algorithm organizes all network nodes in
a graph of local clusters, and computes
efficient schedules for inter- and intra-cluster communication. More precisely:
the computed inter-cluster communication schedule works in contact number of rounds
and can be done in parallel without a harm from the interference to the final result,
similarly as broadcasting intra-cluster schedules;
the convergecast intra-cluster operations, although can be done in parallel, require $\Theta(\Delta)$
rounds, where $\Delta$ is the upper bound on the node 

For the sake of complexity formulas, in the analysis we consider the following parameters:
$n$, $N$, $D$, $\Delta$ where:
$n$ is the number of nodes,
$[N]$ is the range of IDs,
$D$ is the 
diameter
and
$\Delta$ is the upper bound on the
degree of a station in the communication graph
(i.e., roughly, the inverse of the lower bound on the density of nodes).

\paragraph{Messages.}
In general, we assume that a single message
sent in the execution of any algorithm
can carry
a single rumor (in case of multi-broadcast)
and at most logarithmic, in the
size of the ID range $N$, number of control bits.

\paragraph{Knowledge of stations}
Each station knows its own ID, location, and parameters $N$, $\Delta$.
We assume that stations do not know any other information
about the topology of the network at the beginning of the execution of an algorithm.



\section{Technical Preliminaries and SINR-Selectors}
\label{s:selector}

\T
Given a parameter $c>0$, we define 
a partition of the $2$-dimensional space
into square boxes of size $c\times c$ by the grid $G_c$, in such a way that:
all boxes are aligned with the coordinate axes,
point $(0,0)$ is a grid point,
each box includes its left side without the top
endpoint and its bottom side without the right endpoint and
does not include its right and top sides.
We say that $(i,j)$ are the coordinates
of the box with its bottom left corner located at $(c\cdot i, c\cdot j)$,
for $i,j\in \INT$. A box with coordinates
$(i,j)\in\INT^2$ is denoted $C(i,j)$.
As observed in \cite{DessmarkP07,EmekGKPPS09}, the {\em grid} $G_{r/\sqrt{2}}$
is very useful in design of algorithms for geometric radio networks, provided
$r$ is equal to the range of each station.
This follows from the
fact that $r/\sqrt{2}$ is the largest parameter of a grid such that each
station in a box is
in the range of every other station in that box.
In the following, we fix $\gamma=r/\sqrt{2}$, where $r=(1+\eps)^{-1/\alpha}$, and call $G_{\gamma}$
the {\em pivotal grid}. If not stated otherwise, our considerations will
refer to (boxes of) $G_{\gamma}$.

Two boxes $C,C'$ are {\em neighbors} 
if there are
stations $v\in C$ and $v'\in C'$ such that edge $(v,v')$ belongs to the
communication graph of the network.
%
For a station $v$ located in position $(x,y)$ on the plane we define its {\em grid
coordinates} with respect to the grid $G_c$ as the pair of integers $(i,j)$ such that the point $(x,y)$ is located
in the box $C(i,j)$ of the grid $G_c$ (i.e., $ic\leq x< (i+1)c$ and
$jc\leq y<(j+1)c$).
Moreover, $\boxx(v)=C(i,j)$ for a station $v$ with grid coordinates $(i,j)$.
If not stated otherwise, we will refer to grid coordinates with respect
to the pivotal grid.

A (classical) {\em communication schedule} $\mathcal{S}$ of length $T$
wrt  $N\in\NAT$ is a mapping
from $[N]$ to binary sequences of length $T$.
A station
with identifier $v\in[N]$ {\em follows}
the schedule $\m{S}$ of length $T$ in a fixed period of time consisting of $T$ rounds,
when
$v$ transmits a message in round $t$ of that period iff
the $t$th position of
$\m{S}(v)$ is equal to $1$.

A {\em geometric communication schedule} $\mathcal{S}$ of length $T$
with parameters $N,\delta\in\NAT$, $(N,\delta)$-gcs for short, is a mapping
from $[N]\times [0,\delta-1]^2$ to binary sequences of length $T$.
Let $v\in[N]$ be a station whose grid coordinates
with respect to
the grid $G_c$ are equal to $(i,j)$.
We say that $v$ {\em follows}
$(N,\delta)$-gcs $\m{S}$ 
for the grid $G_c$
in a fixed period of time, 
when $v$ transmits a message in round $t$ of that period iff
the $t$th position of
$\m{S}(v,i\mod \delta,j\mod\delta)$ is equal to $1$.
A set of stations $A$ on the plane is {\em $\delta$-diluted} wrt $G_c$, for $\delta\in\NAT\setminus\{0\}$, if
for any two stations $v_1,v_2\in A$ with grid coordinates $(i_1,j_1)$ and $(i_2,j_2)$, respectively,
the relationships $(|i_1-i_2|\mod \delta)=0$ and $(|j_1-j_2|\mod \delta)=0$ hold.
Since we will usually apply dilution to the pivotal grid, it is assumed that all references
to a dilution concern that grid, unless stated otherwise.

Let $\m{S}$ be a classical communication schedule $\m{S}$ wrt $N$ of length $T$,
let $c>0$ and $\delta>0$, $\delta\in\NAT$.
A $\delta$-dilution of $\m{S}$ wrt $(N,c)$ is a $(N,\delta)$-gcs (geometric communication schedule)
$\m{S}'$ of length $T\cdot \delta^2$ defined such such that the bit $(t-1)\delta^2+a\delta+b$
of $\m{S}'(v,a,b)$ is equal to $1$ iff the bit $t$ of $\m{S}(v)$
is equal to $1$. In other words, each round $t$ of $\m{S}$ is
partitioned
into $\delta^2$ rounds $(t,a,b)$
of $\m{S}'$ indexed by pairs $(a,b)\in [0,\delta-1]^2$, such that a station
with grid coordinates $(i,j)$ in $G_c$ is
allowed to send messages only in rounds $(t,i\mod\delta,j\mod\delta)$,
provided schedule $\m{S}$ admits a transmission of this station in its round~$t$.

\begin{figure}
\begin{center}
\epsfig{file=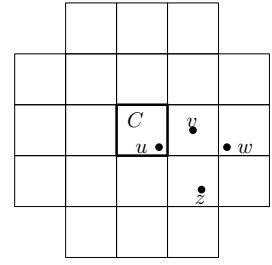, scale=0.8}\vspace*{-2.2ex}
\end{center}
\caption{If $v,w,z$ are in the range are of $u$, then boxes
containing $v,w,$ and $z$ are neighbors of $C$.
}
\label{fig:adjacent}
\vspace*{-1.85ex}
\end{figure}%
Observe that, since ranges of stations are equal to the length
of diagonal of boxes of the pivotal grid, a box $C(i,j)$ can have at most
$20$ neighbors (see Figure~\ref{fig:adjacent}).
We define the set $\DIR\subset[-2,2]^2$ such  that $(d_1,d_2)\in\DIR$ iff
it is possible that boxes with coordinates $(i,j)$ and $(i+d_1,j+d_2)$
can be neighbors.


A set of stations $A$ on the plane is {\em $\delta$-diluted} wrt $G_c$, for $\delta\in\NAT$, if
for any two stations $v_1,v_2\in A$ with grid coordinates $(i_1,j_1)$ and $(i_2,j_2)$, respectively,
the relationships $|i_1-i_2|\mod \delta=0$ and $|j_1-j_2|\mod \delta=0$ are satisfied.

\begin{proposition}\labell{prop:dilsuc}
For each $\alpha>2$ ($\alpha=2$, resp.) there exists a constant $d$ ($d=O(\log n)$, resp.)
such that the following
property is satisfied: \\
Let $A$ be a
$\delta$-diluted, wrt the pivotal grid $G_{\gamma}$,
set of stations on the plane such that $\delta\geq d$ and each box
contains at most one element of $A$. Then, if all elements of $A$ transmit
messages simultaneously in the same round $t$ and no other station is transmitting,
each of them transmit successfully.\\
\end{proposition}

\T\T\T
The correctness of the above proposition simply follows from the fact that
$\sum_{i=1}^{\infty} 1/i^{\sigma}=O(1)$ for each $\sigma>1$,
and $\sum_{i=1}^{n} 1/i=O(\log n)$ for $n\in\NAT$ (these sums
appear when one evaluates maximal possible
interference
generated by
other elements of $A$ to a box containing any $v\in A$).

%

\T
\begin{definition}[SINR selector]
A $(N,\delta)$-gcs $\mathcal{S}$, for $N,\delta\in\NAT$, is a
$(N, \delta,
\Delta,\varepsilon)$-SINR-selector
if for each
set $A$ of $m\leq N$ stations with IDs in the range $[N]$,
located on the plane such that at most $\Delta$ stations are located
in each cell of the pivotal grid, it satisfies the following properties:\\
\textbf{(a)}~At
least $\ep\cdot |A|$ stations from $A$ broadcast successfully during
the execution of $\mathcal{S}$.\\
\textbf{(b)}~Let $B\subseteq A$ be the set of such elements $v\in A$ that $v$ is \textbf{not}
the only element of $\boxx(v)$ which belongs to $A$. Then,
at least $\ep\cdot |B|$ stations from $B$ broadcast successfully during
the execution of~$\mathcal{S}$.
\end{definition}

\comment{ 
A {\em box-aware relocation} of a set of stations $A$ on the plane is
a
placement of elements with IDs from $A$
such
that grid coordinates
of each element of $A$ remain the same as in the original location.

Let $\m{S}$ be a $(N,\delta,\Delta,\varepsilon)$-SINR-selector of length $T$.
We say that $\m{S}$ is {\em placement independent}
if for each set $A$ of $m\leq N$ stations with IDs in the range $[N]$
there exists $A'=\{v_1,\ldots,v_p\}\subseteq A$ of size $p\geq \varepsilon\cdot m$ and
a set of rounds $t_1<\cdots<t_p\leq T$
such
that each station with ID
$v_i$ broadcasts successfully in round $t_i$
in the
execution of $\m{S}$
on each box-aware relocation of $A$.
}

\begin{lemma}\labell{l:geometric-selector}
For each $N\in\NAT$ and $\Delta\leq N$, there exists a 
$(N,\delta,\Delta,1/2)$-SINR-selector $\m{S}$
of size:
$O(\Delta\log^2 N)$ for $\alpha>2$,
and $O(\Delta\log^3 N)$ for $\alpha=2$,
where $\delta=O(\log N)$. Moreover, $\m{S}$ is a $\delta$-dilution of
some classical communication schedule.
\end{lemma}

%
The proof of Lemma~\ref{l:geometric-selector} is given in Appendix~\ref{s:GSelectorA}.
It
is based on the probabilistic method. Consider a randomly
generated classical communication schedule $S$, such that each $v\in[N]$ is chosen to
be a transmitter in round $i$ with probability $\frac{1}{\Delta}$.
The goal is to show that it satisfies the properties
of a {\sselector} after applying $d$-dilution with appropriately chosen $d$.
Unfortunately,
unlike in radio networks, it is not sufficient to prove that a station $v$
is the only (successful) transmitter in the ball of radius $r$ centered at $v$.
Technical challenge here is to deal with interferences going from stations
located in other (even distant) boxes, even thought in expectation there are not many of them.
We deal with this issue by bounding
the probabilities that many boxes may contain more than $2^i$ transmitting stations, for growing
values of~$i$, and calculating probabilities of large interferences using these bounds.
Although for a particular box
it might be the case that the fraction of successfully transmitting stations from this
box is small, we can still prove a global bound on the
number of successfully transmitting stations from the
whole considered set~$A$.

\T\T
\section{Backbone Structure \& Algorithm}
\label{s:backbone}

\T
\input{Backbone.tex}

\T\T\T\T
\section{Applications of Backbone} 
\label{s:alg}

\T
\input{Backbone-Leader-Election.tex}

\T\T\T\T
\input{Conclusions.tex}

\bibliography{references}
\bibliographystyle{abbrv}

\clearpage
\appendix

\begin{center}
{\bf\Large Appendix}
\end{center}

\input{Selector-Exist.tex}

\input{Backbone-Leader-Analysis.tex}

\input{Lower-Leader.tex}

\input{Backbone-Multi-Broadcast.tex}

\end{document}

%% file: Backbone.tex
In this section we describe a method of building a {\em backbone network},
a subnetwork of an input wireless network which acts as a tool
to perform several communication tasks. More precisely, given a network
with a communication graph $G$, we build its subnetwork $H$ which satisfies
the following properties:

\T\T
\begin{enumerate}
\item[$P_1:$]
the stations of $H$ form a connected dominating set of $G$;
\T
\item[$P_2:$]
the number of stations in $H$ is $O(m)$, where $m$ is the size
of the smallest connected dominating set of $G$;
\T
\item[$P_3:$]
each station from $G$ is associated with exactly one its neighbor that
belongs to $H$ (its {\em leader}).
\end{enumerate}

\T\T
Moreover, we design protocols with the following characteristics:

\T\T
\begin{enumerate}
\item[$A_1:$]
a protocol simulating one round of a message passing network on $H$ in a  {\em multi-round},
which consists of constant
number of rounds; more precisely, this protocol makes possible to exchange
messages between each pair of neighbors of $H$ in one multi-round; moreover,
each message received by a station $v\in H$ is successfully broadcasted to
all stations associated with $v$ (i.e., stations for which $v$ is the leader);
\T
\item[$A_2:$]
a protocol which, assuming that each station $v$ has its own message, makes
possible to transmit message of each station to its leader in $O(\Delta)$
rounds, where $1/\Delta$ is the density of the network.
\end{enumerate}

\T\T
Above, by designing a protocol we mean a distributed algorithm during
which each stations
learns its behavior in execution of the designed protocol.

Our algorithms performing the above tasks strongly rely on {\sselector}s. The algorithm
first elects leaders in boxes of the pivotal grid (\LocalLeader), then these leaders acquire
knowledge about stations in their boxes and their neighbors (in the communication graph
(\LocalLeader\ and \NeighborL) and finally a constant number of stations are added in each box
(of the pivotal grid) to these leaders in order to form the network $H$ satisfying stated above properties
(\InterC). Below, we describe these phases of the algorithm in more detail.
We assume that the inverse of density $\Delta$ is known to all stations of a network.
We make a simplifying assumption that if at most one station
from each box of the pivotal grid transmits in a round $t$, then
each such transmission is successful.
Due to Proposition~\ref{prop:dilsuc}, one can achieve this property
using dilution with constant parameter, which does not change the asymptotic complexity
of our algorithm.
We also use the parameter $\delta$ which corresponds to the dilution parameter from
Lemma~\ref{l:geometric-selector}.

\paragraph{\LocalLeader.}
The goal of this phase is to choose the leader in each nonempty box of the pivotal grid (i.e.,
each box containing at least one station of $G$). Each station $v$ has a local variable
$\state(v)$ denoting its state. At the beginning each $\state(v)=\textit{active}$ for each station.
Our algorithm consists of $\log N$ repetitions of $(N,\delta,\Delta,1/2)$-{\sselector} $S$.
After each round $t$ of the algorithm, each station $v$ which can hear a message in
$t$ sent by a station $u\in\boxx(v)$ changes its state to {\em passive}.
After $\log N$ executions of $(N,\delta,\Delta,1/2)$-{\sselector}, each station
$v$ such that $\state(v)=$active changes its state to {\em leader} and becomes
the leader of its box.

Let $M(i)$ be the set of stations which, after the $i$th execution of the {\sselector} $S$,
are in state active and they are located
in boxes which contain at least two stations in state active at this time.
%
\begin{proposition}\labell{prop:lleader:size}
For each $i\in[\log N]$, either $|M(i)|=0$ or $|M(i)|\leq |M(i-1)|/2$.
\end{proposition}
\begin{proof}
Let {\em canonical execution} of $S$ be an execution in which the set of stations
participating during the execution does not change. 
The properties of $S$
imply that the set $X$ of stations from $M(i-1)$ which transmit successfully during
the $i$th execution of $S$ has at least $|M(i-1=|/2$ stations.
Let $X=X_1\cup X_2$, where $X_1$ contains those stations which became passive before
having a successful transmission in $S$, and $X_2=X\setminus X_1$. Note that each
$v\in X_2$ does not belong to $M(i)$, since all stations in $\boxx(v)$ except of $v$
become passive after
a successful transmission of $v$. On the other hand, stations from $X_1$ does not belong
to $M(i)$, since they are in state passive.
Thus, $M(i)\subseteq M(i-1)\setminus X$ and therefore $|M(i)|\leq |M(i-1)|/2$.
\end{proof}
\begin{lemma}\labell{l:lleader}
After $\log N$ application of {\sselector} $S$ during {\LocalLeader},
there is exactly one station in state leader in each nonempty~box.
\end{lemma}
\begin{proof}
Observe that $|M(0)|\leq N$ and therefore, by Proposition~\ref{prop:lleader:size}, $|M(\log N)|=0$.
Thus, there is no box with more than one active station at the end of the algorithm.

It remains to verify that each nonempty box has a station in state active at the end
of the algorithm. However, note that the number of active stations in a box in a particular
round $t$ may decrease
only in the case that at least one station in that box is transmitting in $t$. Thus, this
station remains active. Therefore, it is not possible that a nonempty box has no station
in state active at the end of the algorithm.
\end{proof}

\paragraph{\LocalLearning.}
The goal of this phase is to assure that each leader of a box gets knowledge about all stations in its
box and shares this knowledge with these stations. To this aim, we execute the {\sselector}
$S$ on the set of non-leaders $\log N$ times. And, the leaders send confirmation messages
after each round of the selector if they can hear a message from their boxes. And in turn,
stations which receive confirmation get {\em passive}. Moreover, leaders (and other stations)
store counters of
confirmed stations from their boxes and arrays containing IDs of confirmed stations.
Below, we describe this in more detail.
\begin{algorithm}[H]
	\caption{\LocalLearning}
	\label{alg:LocalLearning}
	\begin{algorithmic}[1]
    \State for each station $v$: if $\state(v)\neq$leader, then $\state(v)\gets$active;
	\State  Repeat $\log N$ times the {\sselector} $S$. In each round, only stations
    in the state {active} can send messages. Each round $t$ of the selector is followed
    by round $t'$, in which each station $v$ runs the following instructions:
        \If{$\state(v)=$leader, $\boxx(v)=C$ and $v$ received a message from $u\in C$ in $t$}
            \State $count\gets count+1$
            \State $set[count]\gets u$
            \State {send a message $(u,count)$, containing the ID $u$, the coordinates of $u$
            and the value of $count$.}
        \EndIf
    After the round $t'$, each station $w$ which receives a message from the leader of its
    box, updates its local copies of $count$ and $set$ appropriately.
    \end{algorithmic}
\end{algorithm}
%
%
The properties of $S$ and an analysis similar to the analysis of the phase {\LocalLeader} directly imply
the following result.
\begin{lemma}\labell{l:LocalLearning}
After {\LocalLearning} each station knows all stations located in its box, stored
in its array $set[1,count]$.
\end{lemma}

\paragraph{\NeighborL.}
Note that leaders of boxes form a ``backbone'' which is a
dominating set of the communication graph.
In fact such a dominating set is at most $|\DIR|+1=21$ times larger
than a smallest dominating set (since each station $v$ can be only in
range of stations located in box $C=\boxx(v)$ and boxes in
directions $(d_1,d_2)\in\DIR$ from $C$).
Our goal is to add
a constant number of stations from each nonempty box in order to guarantee connectivity of
our backbone and make efficient communication inside the backbone possible. To this aim
we first design a communication schedule which gives each station some knowledge about
neighbors of each station in its box. This is achieved by allowing that, given a box $C$,
each station from $C$ sends a message in different round. Thus, each such message is received
by all neighbors of the transmitting station (due to dilution). Moreover, each station
adds locations and IDs of all stations from which it receives messages to
(stored locally) the set of its
neighbors.
%
Given a box $C(i,j)$ and $(d_1,d_2)\in\DIR$, the key information for us is whether and which
stations from $C$ have neighbors in the box $C(i+d_1,j+d_2)$. Therefore, we repeat
the second time this procedure in which each station in each box transmits separately.
This time each station $v$ attaches to its messages $D(v)$, which is
the set of directions from
$\DIR$ describing boxes in which $v$ has neighbors, and for each such direction $(d_1,d_2)$, it attaches
$twin_{(d_1,d_2}(v)$, the smallest ID of such a neighbor\footnote{We do not allow a station to send information
about all its neighbors in order to keep messages small, of size $O(\log N)$.}.
Thanks to that, this information is spread
in the whole $\boxx(v)$.
Given this information, stations ``responsible'' for communication with other boxes
are chosen in each box.
\begin{algorithm}[H]
	\caption{\NeighborL}
	\label{alg:NeighborL}
	\begin{algorithmic}[1]
    \State Each station $v$ sets: $\Gamma(v)=\emptyset$
    \For{each station $v$}
        \For{$i=1,2,\ldots,\Delta$}
            \If{$v=set[i]$}
            \State $v$ sends a message,
            \Else\ if $v$ can hear $u$:$\Gamma(v)\gets\Gamma(v)\cup\{u\}$
            \EndIf
        \EndFor
    \EndFor
    \For{each station $v$}
        \For{$i=1,2,\ldots,\Delta$}
            \State\ \ if $v=set[i]$ then $v$ sends $D(v)$.
            \State\ \ if $v$ can hear $D(u)$ then $v$ stores it.
        \EndFor
    \EndFor
    \For{$(d_1,d_2)\in\DIR$ and $v\in C(i,j)$}
        \State $C'\gets C(i+d_1,j+d_2)$
        \State $nb_{(d_1,d_2)}\gets \{u\in C\,|\, u \mbox{ has neigh.in } C'\}$
        \State if $nb_{(d_1,d_2)}\neq\emptyset$: $s^C_{(d_1,d_2)}\gets\min(nb_{(d_1,d_2)})$
        \State \ \ $r^{C'}_{(-d_1,-d_2)}\gets twin_{(d_1,d_2)}(v).$
        \State \ \ $s^C_{(d_1,d_2)}$ sends a message to $r^{C'}_{(-d_1,-d_2)}$.
    \EndFor
    \end{algorithmic}
\end{algorithm}

\T\T
Note that each station can perform computation from line 10-14
independently using information received in the first
two loops. All stations of type $s^C_{(d_1,d_2)}$ and
$r^C_{(d_1,d_2)}$ are added to the backbone $H$ (except of leaders
of all boxes). The idea
is that the task of $s^C_{(d_1,d_2)}$ (senders) is to send
a message from the box $C$ to the box $C'$ located in direction
$(d_1,d_2)$ from $C$, while the task of $r^{C'}_{(-d_1,-d_2)}$
(receivers)
is to broadcast this message to all stations in $C'$.

\paragraph{Simulation of message passing model inside the backbone.}
Now, we define {\em multi-round} which consists of $|\DIR|\cdot(\delta')^2$
actual rounds of our network, where $\delta'$ corresponds to the parameter
$\delta$ from Proposition~\ref{prop:dilsuc}.

\T
\begin{algorithm}[H]
	\caption{Multi-round}
	\label{alg:multi}
	\begin{algorithmic}[1]
    \For{each station $v$}
        \If{$\state(v)=\text{leader}$} $v$ sends a message;
        \EndIf
    \EndFor
    \For{$(d_1,d_2)\in \DIR$}
        \For{each station $v$}
            \State Round 1: if $(v=s_{(d_1,d_2)}^{\boxx(v)})$:
            \State \ \ send a mess.;
            \State Round 2: if $(v=r_{(-d_1,-d_2)}^{\boxx(v)})$:
            \State \ \ send a mess. received in Round 1;
        \EndFor
    \EndFor
    \end{algorithmic}
\end{algorithm}

\T\T
Note that, since at most one station from each box is sending a message,
each stations is transmitting successfully. Moreover, since each station
from $H$ is sending during a multi-round, each station can (successfully)
send a message to all its neighbors in $H$ during a multi-round.
Below, we formulate a corollary emphasizing the fact that our backbone
provides possibility of communication in a graph whose vertices correspond
to boxes of the pivotal grid and edges connecting boxes which are neighbors
in a network (see definition of neighborhood of boxes and Figure~\ref{fig:adjacent}).

\T
\begin{corollary}\labell{cor:com:box}
Assume that all stations in box $C$ know the same message, for each box of
the pivotal grid.
During a multi-round, a message from
$C$ can be transmitted to all stations in $C'$ for each pair
of stations $(C,C')$ that are neighbors in the underlying network.
\end{corollary}

\T\T
\begin{theorem}\labell{th:backbone}
 The backbone $H$ satisfying properties $P_1-P_3$ and protocols
 satisfying $A_1,A_2$ can be build in time $O(\Delta\polylog(N))$.
\end{theorem}

%% file: Backbone-Leader-Election.tex

In this section two examples of applications of the backbone structure
are presented.
First, we present a leader election algorithm {\GlobalLeader}, where each station
initially knows only its own ID, $N$, and $\Delta$. This algorithm uses the backbone
structure in order to choose local leaders in boxes of the pivotal grid. Then,
using multi-rounds defined for the backbone structure, we can simulate model of ad hoc broadcast
networks without interference. Therefore, our leader election algorithm works in this model,
and the main emphasis is on completing leader election task in short time and with small
size messages, assuming that the diameter of the network is initially unknown to stations.
Finally, using the backbone structure and additional information acquired during leader election
algorithm, we show an efficient algorithm for multi-broadcast.

Below, we describe our leader election algorithm.
First, we execute an algorithm which builds the backbone
$H$ of the communication graph $G$.

After the part {\LocalLearning} of backbone construction,
each station $v$ knows IDs of all stations
in its box $\boxx(v)$ (see Lemma~\ref{l:LocalLearning}). So, assume that the station with the smallest
ID in a box is the leader of that box (known to all stations in the
box, as we observed). And, our goal is to elect the leader of the
network among leaders of boxes.
Then, using Corollary~\ref{cor:com:box}, we can express our
algorithm in terms of communication in the graph, whose
vertices are boxes of the pivotal greed and edges connect
boxes which are neighbors in the underlying network.
Moreover, using properties of multi-rounds, we can guarantee that each two
boxes $C, C'$ which are neighbors exchange messages during one
multi-round. With each nonempty box $C$ we associate the following
variables (which may be stored in local memories of all stations from
the backbone $H$ located in $C$):

\vspace*{-1.5ex}
\begin{itemize}
 \item $\leader(C)$, initially equal to the ID of the local
 leader of $C$;
\vspace*{-1.5ex}
 \item $\state(C)$, initially equal to $\forward$, and in general can be
set to one of the values:
 $\forward,$ $\waitback,$ $\backward,$ $\waitconf,$ $\confirm$, $\stopp$;
\vspace*{-1.5ex}
 \item
 sets $\pred(C)$ and $\suc(C)$, initially empty.
\end{itemize}

\vspace*{-1.5ex}
\noindent
Our leader election algorithm {\GlobalLeader} consists of consecutive repetitions of
multi-rounds.
However, only boxes in some specified states send messages
during a multi-round. Moreover, at the end of each multi-round, each
box updates its local variables. The important issue to deal with is
to guarantee that all boxes know, as early as possible, that leader
is already elected and can finish the computation.

The idea is that each box initially holds a token containing the value of
its leader. The box in state $\forward$
sends the token with ID of its leader to its neighbors and then
moves to the state $\waitback$ in which it is waiting for confirmation that
its leader is elected as the leader in some subnetwork (specified later).
However, if a station receives
a token with smaller ID than its current value $\leader(C)$, then it changes
its value of the leader, gets back to the state forward and repeats the whole process.
At the same time, a directed acyclic graph $G'$ is formed defined by variables $\pred$ and
$\suc$, where edges are $(C_1,C_2)$ for $C_2\in\suc(C_1)$ and an edge $(C_1,C_2)$ denotes
the fact that $C_2$ received the current value $\leader(C_2)$ from~$C_1$.

After receiving confirmation that $\leader(C)$ is the leader of
boxes following it in $G'$, $C$ changes its state to $\backward$ in which it sends a message
confirming the fact that $\leader(C)$ is the leader of the subnetwork ``following'' it. And,
$C$ changes its state to $\waitconf$. However, if a box is the source in $G'$, i.e. it has
no predecessor ($\pred(C)=\emptyset$), reception of $\backward$ message means that
$\leader(C)$ is the leader of the whole network already. And, $C$ gets the state $\confirm$
in which it initiates flooding of the $\confirm$ message informing other stations that
the process of leader election is finished. Other stations enter the $\confirm$ state
after receiving $\confirm$ message from the predecessors.

Now, we describe these ideas in more detail. Our algorithm consists of phases, where
each phase consists of two multi-rounds, each multi-round followed by updates of local
variables of boxes.
\comment{
\begin{algorithm}[H]
	\caption{Multiround 1 in a phase of leader election}
	\label{alg:multi:leader}
	\begin{algorithmic}[1]
	\If{$\state(C)\in\{\forward,\backward,\confirm\}$}
	  $C$ sends $\state(C), \leader(C)$ to its neighbors
	\EndIf
        \end{algorithmic}
\end{algorithm}
}
On the basis of received messages, box $C$ changes its local variables and transmit
information about changes in the second multi-round.
We describe this for
a box~$C$:
\begin{algorithm}[H]
\caption{Multiround 1 of a phase}
\label{alg:states:leader:m1}
\begin{algorithmic}[1]
  \If{$\state(C)\in\{\forward,\backward,\confirm\}$}
    \State $C$ sends $\state(C), \leader(C)$ to its neighbors
  \EndIf
\end{algorithmic}
\end{algorithm}

\T\T
\begin{algorithm}[h]
\caption{Updates after multiround 1}
\label{alg:states:leader:u}
\begin{algorithmic}[1]
  \State $l\gets \min_{C'\in\Gamma(C)}(\leader(C'))$
  \If{$l<\leader(C)$}
    \State $\state(C)\gets \forward$
    \State $\leader(C)\gets l$
    \State $\pred(C)\gets \{C'\in\Gamma(C)\,|\, \leader(C')=l\}$
  \EndIf
  \State \textbf{else:}
  \newpage
  \State \textbf{CASE} ($\state(C)$):
    \State \textbf{$\forward$}: $\state(C)\gets\waitback$
    \State \textbf{$\waitback$}:
      \State $\m{D}\gets \{C'\in\suc(C)\,|\, C'\mbox{ sent }(\backward, \leader(C))\}$
      \If{$\m{D}=\suc(C)$}
        \State\textbf{if} ($\suc(C)=\Gamma(C)$): $\state(C)\gets\confirm$
        \State\textbf{else} $\state(C)\gets\backward$
      \EndIf
    \State \textbf{$\backward$}: $\state(C)\gets \waitconf$
    \State \textbf{$\waitconf$}:
    \If{$C$ received $(\confirm,\leader(C))$ from each $C'\in\pred(C)$} $\state(C)\gets\confirm$
    \EndIf
    \State \textbf{$\confirm$}: $\state(C)\gets \stopp$
    \State \textbf{END CASE}
\end{algorithmic}
\end{algorithm}
\T\T
\begin{algorithm}[H]
\caption{Multiround 2 and updates after it}
\label{alg:states:leader:m2}
\begin{algorithmic}[1]
  \State $C$ sends $\pred(C)$ to each its neighbor;
  \State $\suc(C)\gets\emptyset$
  \For{each $C'$ such that $C\in\pred(C')$}
    \State $\suc(C)\gets \suc(C)\cup\{C'\}$
  \EndFor
\end{algorithmic}
\end{algorithm}
\comment{ 
In the above algorithm, we assume that boxes are identified by IDs of their
original local leaders (chosen during design of the backbone). That is why
we can assume some order on their IDs (which we identify with names of boxes).
}

\T
Formal analysis of our algorithm is postponed to Appendix~\ref{s:app:leader},
where the following theorem is proved.

\T
\begin{theorem}\labell{t:leader:election}
 The algorithm {\GlobalLeader} finishes leader election in $O(D+\Delta\polylog n)$ rounds, where $D$
 is the diameter of the underlying network.
\end{theorem}

We complement the above result by a lower bound, which leaves only a poly-logarithmic gap
in complexity. Interestingly, it also holds for randomized algorithms.
Here, we say that a randomized algorithm works in time $f(n)$, where
$n$ is the size of a network, if it finishes its computation (on
every network) in \tj{$f(n)$ rounds} with probability at least $1-\tau(n)$, where $\tau(n)=o(1)$.

\T
\begin{theorem}\labell{th:lower}
Each (randomized) algorithm
solving leader election problem works in $\Omega(\Delta+D)$ rounds
in the worst case.
\end{theorem}

\T
\noindent
The proof of Theorem~\ref{th:lower} is presented in Appendix~\ref{s:lower}.

Finally, using the backbone structure and the leader election algorithm,
we efficiently solve the multi-broadcast problem.
Initially, $k$ distinct messages are located in various
stations of a network, and the values $N$ and $\Delta$ 
are known to each station (while $k$ and $D$ are not known to them).
Our algorithm uses the backbone
and the edges defined by sets $\pred(C)/\suc(C)$ constructed during algorithm
{\GlobalLeader} to broadcast messages located in each box to all stations located
in that box and to build a tree $T$ in a graph of boxes, whose root is set to the
box containing the leader of the network. Then, all messages are gathered in the root
of $T$ and finally they are broadcasted from the root to the whole tree by a simple
flooding algorithm. The following theorem characterizes efficiency of our solution;
the formal proof of this theorem is presented in Appendix~\ref{s:multi}.

\T
\begin{theorem}\labell{t:multi:broadcast}
There is
an algorithm that finishes multi-broadcast in $O(D+k+\Delta\polylog(n))$
rounds.
\end{theorem}

%% file: Conclusions.tex
\section{Conclusions and Extensions} 
\label{s:conclusion}


\T
It is worth noting that our algorithm constructing the backbone is local in nature,
and it tolerates a small constant inaccuracy in the location data.
In order to avoid difficulties in construction of efficient SINR-selector, it can be
chosen randomly and verified for correctness (it holds with large probability).

We have not specified exact complexities of algorithms that rely on the
size of {\sselector}, because they depend on the actual value of the path loss 
parameter $\alpha$. Since in our applications {\sselector} are applied $\log N$
times, Lemma~\ref{l:geometric-selector} implies that the $\polylog n$ factor
in the complexity formulas of our algorithms is equal to $\log^3 N$
for $\alpha>2$ and $\log^4 N$ for $\alpha=2$.

%% file: Selector-Exist.tex
\def\GSelector{
In this proof all references to (boxes of) the grid concern boxes of the pivotal
grid.
We show the existence of the appropriate geometric communication schedule using
the probabilistic method.
\tj{We prove that there exists a gcs satisfying property (a) from the definition
of {\sselector}. The proof for (b) can be obtained in such a way that all stations
from $A\setminus B$ are ``ignored'' in the consideration and the possibility that
they can generate additional interference for reception of messages send by elements of $B$
can be tackled by assuming that noise is increased to $(1+\varepsilon/2)\cN$. One can
guarantee that the actual interference imposed by elements of $A$ on elements of $B$
is actually at most $\cN\cdot \varepsilon/2$ by increasing the parameter $\delta$
appropriately. Then the original proof for (a) is applied, where the sensitivity
parameter in the definition of SINR is $\varepsilon'$ such that
$$(1+\varepsilon')(1+\varepsilon/2)\beta\cN=(1+\varepsilon)\beta\cN.$$
}

First, we build a classical communication schedule $\m{S}$ randomly.
For each identifier $v\in[1,N]$ and each round $j$,
the $j$th position of $\m{S}(v)=1$
%
with probability $\frac{1}{\Delta}$ and all such random choices
are independent.

Notations introduced in this proof are summarized in Table~\ref{tab:gsel}.
\begin{table}[h]%
\caption{Notations in the proof of Lemma~\ref{l:geometric-selector} (concerning the
choice of $S_t$ for $t\in\NAT$).}
\label{tab:gsel}\centering%
\begin{tabular}{|c|l|}
\hline
Notation & Definition\\
\hline\hline
$m_i$ & the number of stations in the\\
& box $C_i$\\
\hline
$Y_i$ & the number of transmitting \\
&stations
in the box $C_i$\\
\hline
$O_i$ & $C_i$ is interesting,
exactly one\\ &station
from
$C_i\cap A'$ sends \\
&a message  and no other station \\
&from $C_i\cap A$ sends a message\\
\hline
$X$ & the number of boxes with one\\
&station
chosen (i.e. $|\{i\,|\, Y_i=1\}|$)\\
\hline
$Z_j$ & $|\{i\,|\, Y_i\in [2^j,2^{j+1})\}|$; i.e., $X=Z_0$;\\
\hline
\end{tabular}
\end{table}

Let $A$ be any set of $m\leq N$ stations on the plane such that at
most $\Delta$ stations are located in each cell of the grid.
Let $A'\subset A$, $|A'|=m/2$.
Let
$m_i$ be the number of elements of $A$ in the cell $C_i$ for $i\in\INT^2$.

Our goal is to show that, the following events appear simultaneously with
relatively high probability in each round of randomly chosen schedule
(according to the above distribution),
regardless of the distribution of stations on the plane:
\begin{itemize}
\item
there are many boxes (at least $m/(128\Delta)$) in which exactly one station from $A'$
(and no station from $A\setminus A'$) transmits;
\item
there are relatively few boxes (at most $\frac1{2^i}\cdot\frac{m}{\Delta}$) in which the number of
transmitting stations from $A$ is
in the range $[2^i, 2^{i+1})$ for each $i\in[1,\log m]$;
\end{itemize}
The above two facts combined with the dilution procedure with appropriate parameters
should guarantee that there are relatively many stations transmitting successfully
in each round, with high probability. Then, using the probabilistic method, we show
that there exists an appropriate {\sselector}.

We say that a box $C$ is {\em interesting} iff the fraction of the number of all elements
of $A'$ in $C$ to the number of elements of $A$ in $C$ is at least $\frac14$.
The relationship $|A'|=|A|/2=m/2$ implies that there are at least $m/4$ stations
in interesting boxes.

Let us fix a round $t$. Let $O_i$ be the event that: $C_i$ is interesting,
exactly one station from
$C_i\cap A'$ sends a message in a round $t$ and no other station from $C_i\cap A$
sends a message in $t$.
Given an interesting box $C_i$ with $m_i/4\leq m'_i\leq
\Delta$ stations from $A'$, the probability that $O_i$ appears is
$$P(O_i)=\frac{m'_i}{\Delta}(1-\frac{1}{\Delta})^{m_i-1}\geq
\frac{m'_i}{4\Delta}\geq \frac{m_i}{16\Delta}.$$
%
Let $X$ be a random variable denoting the number of cells in
which the event $O_i$ appears in $t$. Then,
$$E(X)\geq\sum_{\{i\,|\,i\mbox{ is interesting}\}}\frac{m_i}{16\Delta}\geq\frac{m}{64\Delta}.$$
Since the choice of transmitting stations in various cells are
independent, we obtain
\begin{equation}\label{e:onecell}
P(X<\frac{m}{128\Delta})<e^{-\frac{m}{8\cdot 64\cdot\Delta}}
\end{equation}
by the Chernoff bound (\ref{eq:ch2}).

Now, our goal is to estimate (i.e., limit from above) the number of
cells which generate large noise, i.e., the cells in which many
stations from $A$ broadcast in a round. Let $Y_i$ be a random variable denoting
the number of broadcasting stations in the box $C_i$ in a fixed round (see Table~\ref{tab:gsel}).
Certainly, $E(Y_i)=\frac{m_i}{\Delta}$. Using the Chernoff bound (\ref{eq:ch1}),
we have
$$
P(Y_i\geq p\cdot E(Y_i))=P(Y_i>p\cdot
\frac{m_i}{\Delta})<\left(\frac{e}{p}\right)^{\frac{pm_i}{\Delta}}
$$
for $p>1$. Therefore
\begin{equation}\label{e:dwadoi}
P(Y_i\geq 2^j)=P(Y_i\geq
2^j\frac{\Delta}{m_i}\cdot\frac{m_i}{\Delta})<\left(\frac{em_i}{2^j\Delta}\right)^{2^j}
\end{equation}
for $j\geq 1$.
Let $Z_j$ denote the number of cells in a round in which the number
of broadcasting stations is in the range $[2^j,2^{j+1})$. Using
linearity of expectation, we obtain
$$E(Z_j)=\sum_i P(Y_i\in [2^j,2^{j+1}))\leq \sum_i P(Y_i\geq 2^j)$$

Thus, using (\ref{e:dwadoi}), we obtain
\begin{equation}\label{e:dwadointer}
\begin{array}{rcl}
E(Z_j)&\leq&\sum_i
\left(\frac{em_i}{2^j\Delta}\right)^{2^j}\\
&=&\left(\frac{e}{2^j\Delta}\right)^{2^j}\cdot\sum_i
m_i^{2^j}\\
&\leq&
\left(\frac{e}{2^j\Delta}\right)^{2^j}\frac{m}{\Delta}\cdot
\Delta^{2^j}\\
&=&\left(\frac{e}{2^j}\right)^{2^j}\cdot\frac{m}{\Delta}
\end{array}
\end{equation}
The inequality in the above calculations follows from the property
$\sum_i m_i^{2^j}\leq \frac{m}{\Delta}\cdot \Delta^{2^j}$ which is satisfied
due to the fact that $0\leq m_i\leq \Delta$ for each $i$ and the function $f(x)=x^{2^j}$
is convex for each $j\geq 0$.

Now, we would like to estimate the probabilities that the number of boxes
with the number of transmitting station in the interval $[2^j,2^{j+1})$ is larger
than $\frac{m}{\Delta}\cdot\frac1{2^j}$.
%
Since the choices of transmitting stations are independent, we can use the Chernoff bound (\ref{eq:ch1}) once
again:
{\small
\begin{equation}\label{e:dwadofinal}
\begin{array}{rcl} P(Z_j\geq
\frac{m}{\Delta}\cdot\frac1{2^j}) &= & P\left(Z_j\geq \left(\frac{e}{2^j}\right)^{2_j}\frac{m}{\Delta}\cdot
\left(\frac{2^j}{e}\right)^{2^j}\frac1{2^j}\right)\\
& < &\left(e\cdot
2^j\cdot\left(\frac{e}{2^j}\right)^{2^j}\right)^{\frac{m}{\Delta}\cdot\frac1{2^j}}\\
&\leq & \left(\frac{e^2}{2^{j/2}}\right)^{m/\Delta}\\ & < &
e^{-m/\Delta}
\end{array}
\end{equation}
}
for $j>10$, thanks to (\ref{e:dwadointer}). Since we concentrate on the asymptotic analysis here, we do not worry
about finding accurate upper bound on the above probability for $j\leq 10$. Even if
this probability is large, we can diminish its impact by applying
$d$-dilution for
some constant $d$.

Our goal is to show that, with (relatively) high probability, there are many
boxes with exactly one transmitting stations from $A'$ and this station transmits successfully
(i.e., its message is heard by all its neighbors, it is not scrambled by noise
from other boxes)
First, let us assume that our input set $A$ is $d$-diluted for some $d\in\NAT$.
%
In such settings, we define a {\em noise} heard in an nonempty box $C$ in a round $t$
is the maximum (over all positions in $C$) of the power of all stations
transmitting in $t$, except stations located in $C$.
Given a cell $C$ with $p$ stations sending a message, we would
like to calculate the overall amount of noise generated by $C(k,j)$,
heard in all nonempty boxes (i.e., the sum of noises heard in all
boxes $C'(k',j')$ except of $C(k,j)$ such that their grid coordinates
satisfy $|k'-k|\mod d=0$ and $|j'-j|\mod d=0$).
This noise is
\begin{equation}\label{e:dimension}
\begin{array}{rcl}
O\left(\sum_{i=1}^{\infty}\frac{p}{(di)^{\alpha}}\cdot
(8i)\right)&=&O\left(8p\sum_{i}d^{-\alpha}i^{1-\alpha}\right)\\
&=&O\left(\frac{p}{d^{\alpha}}\right)
\end{array}
\end{equation}
for $\alpha>2$, since the number of boxes $C'(k',j')$ such that
$\max(|k'-k|,|j'-j|)=di$ is $\leq 8i$.
Assume that the number of boxes with the number of transmitting
stations in the range $[2^j, 2^{j+1})$ is at most $\frac{m}{\Delta 2^j}$ for each
$j\leq \log m$. Then, the overall noise heard in the network
(i.e., the sum of noises in non-empty boxes)
is
\begin{equation}\label{e:dimension2}
O\left(\sum_{j=1}^{\log m}\frac{2^{j+1}}{d^{\alpha}}\cdot \frac{m}{\Delta 2^j}\right)=O\left(\frac{m}{\Delta}\cdot\frac{\log m}{d^{\alpha}}\right)
\end{equation}
Therefore, if $d>(c\log n)^{\frac{1}{\alpha}}$ for large enough
constant $c$, the overall noise is at most $\frac{m}{4\cdot 128\cdot \Delta}$. In such
case, the number of cells with noise greater than $1/2$ is
at most $\frac{m}{2\cdot 128\cdot\Delta}$.
According to (\ref{e:onecell}) and (\ref{e:dwadofinal}), the
probability that
\begin{enumerate}
\item[(a)]
there are at least $\frac{m}{128\Delta}$ boxes in which exactly one station transmits; and
\item[(b)]
there are at most $\frac{m}{2\cdot 128\cdot \Delta}$ boxes in which the noise can prevent a successful
transmission
\end{enumerate}
is at least
$$1-e^{-m/\Delta}\log m\geq 1-e^{-m/2\Delta}$$
for $m\geq 2\Delta\log\log N$ and $d$-diluted configuration for
$d>(c\log m)^{\frac{1}{\alpha}}$. Therefore the probability that
there are at least $\frac{m}{128\Delta}$ successful transmissions
is at least $1-e^{-m/2\Delta}$ as well.

Now, we are ready to show that there exists $(N,d,\Delta,1/2)$-{\sselector}
of size $O(\Delta\log N)$ for $d$-diluted
configurations (inputs), where $d=O(\log^{1/\alpha}m)$.
Contrary assume that it is not the case.
This assumption implies that the following event $E$ appears with probability $1$:
\begin{quote}
there
exists a set $A$ of size $m\leq N$ and its subset $A'$ of size $m/2$
such that no element of $A'$ broadcasts successfully in
$O(\Delta\log N)$ rounds.
\end{quote}
Our goal is to obtain contradiction by showing that this probability is smaller than $1$.
Therefore, we would like to find a reasonable upper bound on the number of possible
choices of $A$ and $A'$. However, since $A$ is determined not only by IDs of stations
but also by their positions
on the plane, the actual number of choices is infinite for fixed $N$. Fortunately,
our analysis does not rely on the actual positions of stations -- it is only important
which box a station belongs to. Moreover, our considerations are immune on shifts of
configurations. Finally, since the noise decreases with increase of
distance, we may assume that all stations are located in a matrix of $m\times m$ boxes
which is $d$-diluted (i.e., consecutive rows/columns of the matrix are
separated by $d=O(\log^{1/\alpha}m)$ empty rows/columns).
Therefore, the probability of the above defined event $E$
is at most
$$\begin{array}{rcl}
\sum_{m=1}^N {N\choose m}\cdot m^{2m}\cdot {m\choose m/2}
\left(e^{-\frac{m}{2\Delta}}\right)^{c\Delta\log N}&\leq&\\ \sum_{m=1}^N N^m\cdot
m^m\cdot (2e)^{m/2}\cdot e^{-cm\log N/2}&<& 1
\end{array}$$ for large enough
constant $c$, since:
\begin{itemize}
\item
there are ${N\choose m}$ possible choices of IDs for $A$ of size $m$;
\item
each element of $A$ can occupy one of $m^2$ boxes;
\item
$A'$ of size $m/2$ can be chosen in $m\choose m/2$ ways;
\item
the probability that no element of $A'$ transmits successfully in a round
is at most $e^{-\frac{m}{2\Delta}}$.
\end{itemize}
Note that, up to this point, we build a classical communication schedule since
the random choices depended only on IDs.
In order to transform this classical communication schedule working
for $d$-diluted instances into
a {\sselector}, 
it is sufficient to apply the dilution procedure with parameter
$d=O((\log n)^{\frac1{\alpha}})$ for arbitrary instances of the problem
(i.e., sets $A$) and $\m{S}$.
This gives a $(N,d,\Delta,\frac12)$-{\sselector} of size
$O(\Delta\log N\cdot (\log N)^{\frac{2}{\alpha}})=O(\Delta\log^2
N)$, provided $\alpha>2$.
Recall that we proved that these bounds hold under the assumption
that $m\geq 2\Delta\log\log N$. If $m<2\Delta\log\log N$, we can
apply a standard $(n,k,\varepsilon)$-selector of size
$O(\Delta\log N\log\log N)$ from \cite{BonisGV03} (and
interleave it with the selector provided above).

Finally, if $\alpha=2$ then the noise calculated in (\ref{e:dimension}) is
$O(p\log N/d^{\alpha})$ (since $\sum_{i=1}^m=O(\log N)$ and the overall noise calculated in (\ref{e:dimension2})
is $O(\frac{m}{\Delta}\cdot \frac{\log^2N}{d^{\alpha}}$. Therefore, in order
to bound the noise from above by $\frac{m}{32\Delta}$, $d=O(\log n)^{\frac{2}{\alpha}}$
is sufficient. Using $d$-dilution, we obtain a {\sselector} of size
$O(\Delta\log N\cdot d^2)=O(\Delta\log^3 N)$.

\comment{ 
Placement independence of the above construction as well as the fact
that it builds the $\delta$-dilution of a classical communication schedule
follow directly from the construction.
}

}

\section{Existence of Efficient Geometric Selector: Proof of Lemma~\ref{l:geometric-selector}}
\labell{s:GSelectorA}

\paragraph{Technical preliminaries.}
The following widely-known facts will be used later in this section.

\begin{lemma}
The following equations hold:
\begin{enumerate}
\item
$(1-1/p)^p\geq 1/4$ for $0<p\leq 1/2$;
\item
$(1-1/p)^p\leq 1/e$ for $0<p<1$;
\item
$\sum_{i=1}^{\infty} 1/i^{\alpha}=O(1)$ for each $\alpha>1$;
\item
$\sum_{i=1}^{n} 1/i=O(\log n)$ for $n\in\NAT$.
\end{enumerate}
\end{lemma}

\begin{lemma}\labell{lm:Chernoff}(Chernoff bounds).
Let $X_1,\ldots,X_n$ be independent Bernoulli random variables, and let $X=\sum_{i=1}^nX_i$.
Then,
\begin{enumerate}
\item
For any $\zeta>1$ and $\mu\geq E(X)$,
\vspace*{-1ex}
\begin{equation}\label{eq:ch1}
P(X\geq \zeta\mu)<\left(\frac{e}{\zeta}\right)^{\zeta\mu}
\end{equation}
\item
For any $0<\delta<1$ and $\mu\leq E(X)$,
\vspace*{-1ex}
\begin{equation}\label{eq:ch2}
P(X<(1-\delta)\mu)<e^{-\frac{\delta^2}{2}\mu}
\end{equation}
\end{enumerate}
\end{lemma}

We resume the proof of Lemma~\ref{l:geometric-selector}.
\GSelector

%% file: Backbone-Leader-Analysis.tex
\section{Analysis of leader election algorithm}\label{s:app:leader}

In this section we formally analyze algorithm {\GlobalLeader}, that is,
we prove Theorem~\ref{t:leader:election}.

Given a box $C$, we define $span(C)$ as a directed graph of boxes, where
$(C_1,C_2)$ is an edge iff $C_1$ and $C_2$ are neighbors and
$\dist(C,C_2)=\dist(C,C_1)+1$. Note that $span(C)$ is an acyclic graph.
Moreover, we define $follow_C(C_1)$
as the set of vertices accessible from $C_1$ in $span(C)$, i.e.,
$C'$ belongs to $follow_C(C_1)$ iff there exists a path from $C_1$
to $C'$ in $span(C)$.

Let $l_0(C)$ denote the initial value of $\leader(C)$ for a box $C$,
i.e., ID of the smallest station located in $C$.

The following properties follow from the implementation of Algorithms~\ref{alg:states:leader:m1}-\ref{alg:states:leader:m2}.
\begin{proposition}\labell{prop:leader:basic}
At each step of the algorithm {\GlobalLeader} and for each
boxes $C_1,C_2$, it holds: $\leader(C_1)\leq \leader_0(C_1)$;
$C_1\in\pred(C_2)$ iff $C_2\in\suc(C_1)$.
Moreover, if $C_1\in\pred(C_2)$, then $(C_1,C_2)$ is an edge
in $span(C)$, where $\leader_0(C)=\leader(C_1)$ (i.e., the
current leader known to $C_1$ is located in $C$).
\end{proposition}

Now, we state and prove that a box $C'$ can get the state $\backward$
with a particular value of the leader only if all elements of $follow_C(C')$
were in this state before, where $C$ is a box which contains the station
with ID equal to $\leader(C')$.
\begin{proposition}\labell{prop:follow}
Let $C'$ be a box such that $\leader(C')=\leader_0(C)$ and $\state(C')=\waitback$
at the beginning of (multi-)round $t$. Then, the box $C'$ changes its state to $\backward$
or $\confirm$ in round $t$ iff each element  $C''\in follow_C(C')$ satisfies:
$\state(C'')=\backward$ and $\leader(C'')=l_0(C)$ at $t$ or earlier.
\end{proposition}
\begin{proof}
We prove this statement by induction wrt to the size of $follow_C(C')$.
If $|follow_C(C')|=0$, then $\suc(C')=\emptyset$ and therefore the
statement is true (see lines 10-13).

Now assume that the statement holds for each $C''$ such that $|follow_C(C'')|\leq f$
and let $|follow_C(C')|=f+1$. The box $C'$ moves to $\backward$ or $\confirm$
when each element of $\suc(C')$ is in state $\backward$. Since $\suc(C')\subseteq follow_C(C')$,
$|follow_C(C'')|\leq f$ for each $C''\in\suc(C')$. Thus, by inductive assumption,
each element of
$$follow_C(C')=\bigcup_{C''\in\suc(C')}\{C''\}\cup follow(C'')$$
satisfies
$\state(C'')=\backward$ and $\leader(C'')=l_0(C)$ at $t$ or earlier.
\end{proof}
The following lemma can be proved by simple induction wrt to $k$.
\begin{lemma}\labell{l:leader:progres}
 For each nonempty box $C$, the value of $\leader(C)$ after $k$ phases
 is equal to the minimum of values $l_0(C')$ over all boxes $C'$ in distance
 at most $k$ from $C$ in the graph $G'$.
\end{lemma}

\begin{lemma}\labell{l:leader:correct}
 If $\state(C)\in\{\confirm,\stopp\}$ then $\leader(C)$ is equal to the smallest
 ID of stations in the network, i.e., $\min_{C'}(\leader_0(C'))$.
\end{lemma}
\begin{proof}
Let $l=\min_{C'}(\leader_0(C'))$ and let $C$ be the box which
contains the smallest leader, i.e., $\leader_0(C)=l$.
For the sake of contradiction, assume that there exists $C'\neq C$ and
round $t$ such that $\state(C')=\confirm$ and $\leader(C')\neq l$ in
round $t$. Let $t$ be the smallest such round and let $C'$ be the box
satisfying the above properties at round $t$.
According to the algorithm, as no box is in the state $\confirm$ before $t$,
$C'$ moves to the states $\confirm$ since
all its neighbors are in state $\backward$ with leader equal to $\leader(C')$
(see lines $10-13$).
This in turn implies that $\pred(C')$ is empty and thus $\leader(C')$ is equal
to its initial leader $\leader_0(C')$.
Then, Proposition~\ref{prop:follow} implies that each box in $follow_{C'}(C')$
was in state $\backward$ with leader equal to $\leader_0(C'))$ at $t$ or earlier.
However, $follow_{C'}(C')$ is equal to the whole graph. Now, we get a contradiction,
since box $C$ cannot change $\leader(C)$ to $\leader_0(C')$, since $\leader_0(C)<\leader_0(C')$.
\end{proof}
The following lemma implies that each box can finish algorithm after it enters
state $\stopp$.
\begin{lemma}\labell{l:leader:stop}
 If $\state(C)=\stopp$ after some phase then $\state(C)$ will not change in
 the following phases.
\end{lemma}
\begin{proof}
Lemma~\ref{l:leader:correct} implies that a box $C$ enters state $\confirm$
only when $\leader(C)=\min_{C'}(\leader_0(C'))$. Thus, after entering $\confirm$,
the box $C$ does not change its leader $\leader(C)$ and therefore it can only
change state to $\stopp$. 
Thus, $C$ can not
change the state $\stopp$, since its leader $\leader(C)$ is the smallest in the network then,
while it can leave $\stopp$ only after receiving a smaller value of the leader.
\end{proof}

\tj{The following lemma combined with Lemma~\ref{l:leader:stop} directly implies the result stated in Theorem~\ref{t:leader:election}.}
\begin{lemma}\labell{l:leader:time}
Let $\leader_0(C)=\min_{C'}(\leader_0(C'))$ and let $D=\max_{C'}(\dist(C,C'))$. Then,
each box is in state $\stopp$ after $3D+1$ phases.
\end{lemma}
\begin{proof}
The intuition is that the leader of $C$ is spread over the whole graph
in a wave which gets vertices in distance $d$ from $C$ after $d$ rounds.
Then, the wave of messages with the state $\backward$ goes back
to $C$. Finally, $C$ initiates the third wave of transmissions with
$\confirm$ message. Formally, one can show by induction the following
statements:
\begin{itemize}
\item
$\leader(C')=l_0$ for each box $C'$ after $t$ phases
for each $t\geq D$ (by Lemma~\ref{l:leader:progres});
\item
$\state(C')\in\{\backward,\waitconf,\confirm,\stopp\}$ after
$D+k+1$ phases for each box $C'$ such that $\dist(C,C')\geq D-k$;
(note that a box $C'$ such that $follow_C(C')=\emptyset$ moves
to the state $\backward$ directly after receiving $l_0$);
\item
$\state(C)=\confirm$ after $2D+1$ phases;
\item
$\state(C')=\confirm$ after $2D+k+1$ phases for each $C'$ such
that $\dist(C,C')\leq k$.
\end{itemize}
\end{proof}

%% file: Lower-Leader.tex
\section{Lower Bound for Leader Election}
\label{s:lower}

In this section we prove the lower bound on time complexity
of leader election which holds also for randomized algorithms, i.e.,
we prove Theorem~\ref{th:lower}.

First, we concentrate on deterministic algorithms and we describe ideas
leading to the actual formal proof. Let $A$ be an arbitrary algorithm
for leader election problem.

Recall that $r=(1+\eps)^{-1/\alpha}$ is the largest distance between two stations
$u,v$ such that if $u$ is sending a message, $v$ can hear this message (provided
interferences of other stations are small enough). Let us fix the value of $\Delta$.
Now, we define $\m{F}_{d,\Delta}$ (or simply $\m{F}$),
a family of networks 
on which we analyze the algorithm
$A$, where exact value of the constant $d$ 
will be specified later. Let:
\begin{itemize}
 \item
$Z_0$ and $Z_1$ be two sets of points on the plane,
 where $Z_i=\{z_{i,1},\ldots,z_{i,2\Delta-1}\}$, for $i\in[0,1]$,
 the coordinates of $z_{0,j}$ are $((j-1)\cdot d,0)$ and
 the coordinates of $z_{1,j}$ are $((j-1)\cdot d,r)$.
 \item
 ID of $z_{0,j}$ is $j$ and ID of $z_{1,j}$ is $2\Delta+j$
 for $j\in[2\Delta-1]$;
 \item
 $d$ satisfies the inequality:
 $2\Delta d\leq r$.
\end{itemize}
That is, given a network with stations on positions from $Z_0\cup Z_1$, the sets of
edges of its
communication graph consists of two cliques $Z_0$ and $Z_1$ and the sets
$\{\{z_{0,j},z_{1,j}\}\,|\, j\in[2\Delta-1]\}$.

Let $F(X_0,X_1)$ for $X_i\subseteq Z_i$ be a network which consists of stations
located on positions from $X_0\cup X_1$ and with IDs assigned to these points.
For $X\subseteq Z_i$, we define $p(X)=\{j\,|\, z_{i,j}\in X\}$.
The family $\m{F}$ consists of networks $F(X_0,X_1)$ such that:
\begin{itemize}
 \item
  $X_i\subseteq Z_i$ and $|X_i|=\Delta$ for $i\in[0,1]$;
 \item
  $p(X_0)\cap p(X_1)=1$.
\end{itemize}
That is, the communication graph of $F(X_0,X_1)\in\m{F}$ forms two clicks ($X_0$ and $X_1$)
and one edge $\{z_{0,j},z_{1,j}\}$, where $j$ is the only element of $p(X_0)\cap p(X_1)$
(see Figure~\ref{fig:lower}).
\begin{figure}
\begin{center}
\epsfig{file=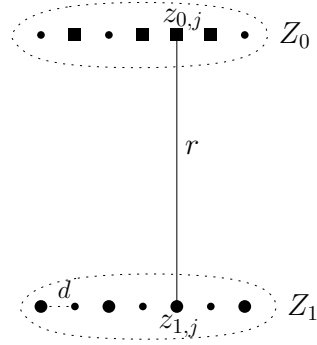, scale=0.8}\vspace*{-2.2ex}
\end{center}
\caption{{\small An example of a network $F(X_0,X_1)\in\m{F}$ for $\Delta=4$. The positions of elements
from $X_0$ and $X_1$ are marked by large squares and large boxes, respectively.}
}
\label{fig:lower}
\vspace*{-1.85ex}
\end{figure}%
As we show later, one can choose the parameter $d$ such that the following properties
are satisfied for each network $F(X_0,X_1)\in \m{F}$ and each time step of any communication
algorithm (see Proposition~\ref{prop:p1p2}):
\begin{description}
 \item[(P1)]
  If at least one station from $X_i$ is sending a message in round $t$, then the result
  of round $t$ for each $v\in X_i$ is independent of the fact whether (and how many)
  stations from $X_{1-i}$ are transmitting messages in $t$.
 \item[(P2)]
  If more than one station from $X_i$ is sending a message in round $t$, then this
  message is not received by any station from $X_{1-i}$.
\end{description}
By {\em result} of a round for a station $v$ we mean here the fact whether $v$ receives or does
not receive a message in the round and, in the former case, the actual message received
by $v$.
Assuming that (P1) is satisfied, each station from $X_i$ for $i\in[0,1]$
is not able to distinguish
between the situation that the network consists of $X_i$ and the situation that the
network consists of $X_0\cup X_1$
until the step in which a station from $X_i$ should receive a message
from $X_{1-i}$ (according to the specification of the algorithm).
On the other hand, by (P2), a message from
$X_{1-i}$ can be received by a station from  $X_{i}$ if exactly one station
from $X_{1-i}$ is transmitting in a round and no station from $X_i$ is transmitting in
that round. Moreover, for each $X_i\subseteq Z_i$, and each $x\in X_i$,
our family $\m{F}$ contains
a network which contains $X_i$ and such $X_{1-i}$ that $x$ is the only element of $X_i$
whose message can be received by a station of $X_{1-i}$. Therefore, in the worst case,
an algorithm can not gain a knowledge whether $X_{1-i}$ is empty or not until each
station from $X_i$ sends a message as the only element of $X_i$. However, this requires
$\Omega(\Delta)$ rounds. Below, we formalize  this intuition, but first we
show that one can
choose $d$ such that (P1) and (P2) are satisfied.

\begin{proposition}\labell{prop:p1p2}
 For each $\Delta$, one can choose $d$ such that properties (P1) and (P2) are satisfied.
\end{proposition}
\begin{proof}
Observe that the distance between $z_{i,j_1}$ and $z_{(1-i),j_2}$ for $i\in[0,1]$ and
$j\in [2\Delta-1]$ is smaller than $r+2\Delta$. Let $F(X_0,X_1)$ and
$\{j\}=p(X_0)\cap p(X_1)$ (i.e, $\{z_{0,j},z_{1,j}\}$ is the only edge between $X_0$
and $X_1$ in the communication graph). If more than one element of $X_{1-i}$ is sending
a message the value $SINR(z_{1-i,j},z_{i,j},\m{T})$ is smaller than
$$\frac{\frac{1}{r^{\alpha}}}{1+\left(\frac1{r+2\Delta\cdot d}\right)^{\alpha}}$$
which is smaller than one provided $d<\left( (1/\varepsilon)^{1/\alpha}-(1/(1+\varepsilon))^{1/\alpha}\right)/(2\Delta)$.
(Recall that $r=(1/(1+\varepsilon))^{1/\alpha}$.)
This proves (P2).

Now, we prove (P1). First, assume that no station from $X_1$ is transmitting a message and
compute SINR
for $x\in X_0$ and a round $t$. Then, the closest to $x$ transmitting station is in distance $id$
for $i\in[1,2\Delta-1)$,
the remaining transmitting stations located to the left of $x$ are in distances
$\{j\cdot d\,|\, j\in J_1\}$ from $x$ and
the remaining transmitting stations located to the left of $x$ are in distances
$\{j\cdot d\,|\, j\in J_2\}$ from $x$, where $J_1,J_2\subseteq [2\Delta-1)$.
Let $$\begin{array}{rcl}
S_0&=&\frac1{(id)^{\alpha}}-\beta(\cN+\sum_{j\in J_1\cup J_2}\frac1{(jd)^{\alpha}})\\
&=&
-\beta\cN+d^{-\alpha}\left(\frac1{i^{\alpha}}-\beta\sum_{j\in J_1\cup J_2}\frac1{j^{\alpha}}\right).
\end{array}$$
Then, $x$ receives a message iff $S_0>0$.

Let $\m{I}=\{(i,J_1,J_2)\,|\, i\in[2\Delta], J_1,J_2\subseteq[2\Delta)\}$.
Let
$$\begin{array}{rcl}
C_0=\{(i,J_1,J_2)\,|\, (i,J_1,J_2)\in\m{I}, \\ \frac1{i^{\alpha}}-\beta\sum_{j\in J_1\cup J_2}\frac1{j^{\alpha}}=0\}
\end{array}$$
and $c_2=\min_{(i,J_1,J_2)\in \m{I}\setminus C_0}(|\frac1{i^{\alpha}}-\beta\sum_{j\in J_1\cup J_2}\frac1{j^{\alpha}}|)$.

Note that stations from $X_1$ can add to the amount of interferences/noise $S_0$
at most the factor $-\beta\cN(1+\varepsilon)\Delta$. Let
$c_1=-\beta\cN(1+(1+\varepsilon)\Delta)$, i.e., $c_1$ accumulates the noise and the largest
possible interference generated by stations from $X_1$.
Thus, $x$ receives a message iff $S_1>0$, where
$S_1$ satisfies
$S_0\geq S_1\geq-c_1+d^{-\alpha}\cdot\left(\frac1{i^{\alpha}}-\beta\sum_{j\in J_1\cup J_2}\frac1{j^{\alpha}}\right)$.
%
Finally, let us choose any $d$ satisfying the inequality $d^{-\alpha}c_2\geq 2c_1$, say
$d=\left(\frac{c_2}{2c_1}\right)^{1/\alpha}$.

Our goal is to show that the additional interference does not change the fact
whether $x$ receives the message. Consider two cases:
\begin{enumerate}
 \item $\frac1{i^{\alpha}}-\beta\sum_{j\in J_1\cup J_2}\frac1{j^{\alpha}}\leq 0$\\
 Then, $S_0<0$ and therefore $S_1\leq S_0<0$
\item $\frac1{i^{\alpha}}-\beta\sum_{j\in J_1\cup J_2}\frac1{j^{\alpha}}> 0$\\
Then $S_1\geq -c_1+d^{-\alpha}c_2=c_1>0$ and $S_0\geq S_1>0$.
\end{enumerate}
\end{proof}

Now, we formalize the idea that the algorithm $A$ needs linear time wrt to $\Delta$
in order to finish leader election.

\comment{For the sake of further generalization of the
result for randomized algorithms, we present our ideas in a formal and possibly
a bit (unnecessarily) complicated way. First, we define a bipartite undirected graph
$H(W_0\cup W_1, E)$ such that:
\begin{itemize}
 \item $W_i=\{X\,|\, X\subseteq Z_i, |X|=\Delta\}$;
 \item $E=\{\{X_0,X_1\}\,|\, X_i\in W_i, |p(X)\cap p(Y)|=1\}$.
\end{itemize}
That is
\begin{enumerate}
 \item there is a $1-1$ mapping between
 edges of $H$ and networks in $\m{F}$:
 an edge $\{X_0,X_1\}$ of $H$ corresponds to a network $F(X_0,X_1)\in\m{F}$;
 \item

\end{enumerate}
}

For the sake of contradiction assume that $A$ performs leader election in time
$\leq \Delta/3$.
Given a network $F(X_0,X_1)\in\m{F}$ we say that a station $x\in X_i$
{\em serves} $F(X,Y)$ if there is a round $t\leq \Delta/3$ such that:
\begin{itemize}
 \item no message from $X_j$ is received by any station of $X_{1-j}$ in rounds
$1,2,\ldots,t-1$ for $j\in[0,1]$;
 \item a message transmitted by $x$ in round $t$ is received by a station from $X_{1-i}$.
\end{itemize}
The set $X_0$ {\em serves} a network $F(X_0,X_1)$ if some $x\in X$ serves $F(X_0,X_1)$.
A network $F(X_0,X_1)$ is {\em served} if it is served by $X_0$ or by $X_1$.
\begin{proposition}\labell{prop:serv}
 If $A$ finishes leader election in $\Delta/3$ rounds then, for each network $F(X_0,X_1)$,
 either $X_0$ or $X_1$ serves this network.
\end{proposition}
\begin{proof}
Assume that there is a network $F(X_0,X_1)$ which is not served by $X_0$ nor by $X_1$.
W.l.o.g. assume that the leader choosen by $A$ in $F(X_0,X_1)$ belongs to $X_0$. Then,
after $\Delta/3$ rounds, the stations from $X_1$ cannot distinguish between the network
$F(X_0,X_1)$ and the network which contains merely stations from $X_1$. Thus, no leader
is elected in the latter case up to the round $\Delta/3$.
\end{proof}

Observe that, according to the construction of the family $\m{F}$ and properties
(P1), (P2):
\begin{itemize}
 \item each set $X_i\subseteq Z_i$ of size $\Delta$ 
 serves at most $\Delta/3$ networks;
 \item
 there are ${2\Delta-1 \choose \Delta}\cdot\Delta$ networks in $\m{F}$;
 \item
 there are $2\cdot {2\Delta-1 \choose \Delta}$ sets $X$ such that
 $F(X,Y)\in \m{F}$ or $F(Y,X)\in\m{F}$ for some $Y$.
\end{itemize}
Therefore, by a simple counting argument, there are
at most $2\cdot{2\Delta-1\choose \Delta}\cdot\Delta/3$ networks which are served
and therefore at least
${2\Delta-1 \choose \Delta}\cdot\Delta/3$ networks  $F(X_0,X_1)\in\m{F}$ such
that $F(X_0,X_1)$ is not served 
Thus, by Proposition~\ref{prop:serv},
we get a contradiction with the assumption that $A$ finishes leader election in
$\Delta/3$ rounds.

Now, consider a randomized algorith $A'$.
The probability that the leader is elected in $\Delta/3$ rounds in a network
$F(X,Y)\in\m{F}$ is bounded from
above by the probability that the network $F(X,Y)$ is served.
Our goal is to show that there exists a network $F(X,Y)$ which is not served
with constant probability.
If this is the case, $A'$ finishes its computation in $\geq \Delta/3$ rounds
with constant probability and therefore Theorem~\ref{th:lower} holds.
Let $\Upsilon$ be the space of all possible probabilistic choices made by $A'$.
Let $f=|\m{F}|$ be the number of networks in $\m{F}$.
For a random sequence $\mu\in\Upsilon$, let $n_{\mu}$ be the number of networks
from $\m{F}$ served by the algorithm $A'$. Our proof for deterministic algorithms
shows that $n_{\mu}< f/3$ for each $\mu\in\Upsilon$. Then, the expected number
of networks served by $A'$ is at most
\begin{equation}\label{e:lower}
\sum_{\mu\in\Upsilon} \mbox{Prob}(\mu)\cdot n_{\mu}< \frac{f}{3}.
\end{equation}
Now, we can easily show that there exists a network $F(X,Y)\in\m{F}$ which is served
with probability smaller than $1/3$. Assume that it is not the case
(i.e., each networks in $\m{F}$ is served with pbb larger or equal to $1/3$) and let
$I(X,Y)$ be the indicator random variable equal to $1$ if $F(X,Y)$ is served and
$0$ otherwise.
Then, the expected number of served networks is at least
$$\sum_{F(X,Y)\in\m{F}} \mbox{Prob}(I(X,Y)=1)\geq f/3$$
which contradicts (\ref{e:lower}).

%% file: Backbone-Multi-Broadcast.tex
\section{Multi-broadcast problem}\label{s:multi}
In this section we describe the algorithm {\MultiAlg} which solves the multi-broadcast problem
efficiently with help of the backbone structure from Theorem~\ref{th:backbone}.
That is, we prove Theorem~\ref{t:multi:broadcast}.

Recall that we assume that a message sent by a station in a round
may contain at most one of original messages which should be
disseminated to all station and $O(\log N)$ additional ``control'' bits.

First, we run the algorithm which builds the backbone $H$. Using the
data structures build during this algorithm, we can deliver all messages
from a box $C$ to all stations in $C$ for each nonempty box $C$, in
$O(\Delta+k)$ rounds. We describe the remaining part of our algorithm
{\MultiAlg} in terms of dissemination of messages
in a graph of boxes (we can think that messages are already collected ``in
boxes'').

First, we execute the algorithm {\GlobalLeader} from the previous section
which elects the leader in the whole network in time $O(D)$ (provided the
backbone is constructed before). However, we introduce the following
modification to the algorithm.
\comment{
The counter $c$ is associated to each message sent by a box in the
state $\forward$ which is modified according to the following rules:
initially each box $C$
sets its counter $c(C)$ to $0$ and each time it changes the value $\leader(C)$
to the received value $\leader(C')<\leader(C)$, it sets $c(C)$ to $c(C')+1$,
where $c(C')$ is the value of the counter of $C'$ received together with
the value $\leader(C')$.
}
After finishing the leader election algorithm, each box $C$
(except the box containing the leader) chooses one element $C'$
of $\pred(C)$ (e.g., the one with the smallest value of $\leader_0(C')$)
as its predecessor and sets $\pred(C)$ to $\{C'\}$. Then, one multi-round
is executed in which each station $C$ informs its neighbors about the new
value of $\pred(C)$. After the multi-round, each box $C$ changes $\suc(C)$
appropriately (i.e., $\suc(C)$ is equal to boxes $C'$ for which $C$ is the
only element of $\pred(C)$).

Let $C$ be the box containing the leader of the whole network.
Since the original graph defined by sets $\pred(C)$ and $\suc(C)$ contains
exactly such edges $(C_1,C_2)$ that $\dist(C,C_2)=\dist(C,C_1)+1$ and $C_1,C_2$
are neighbors, we obtain a tree $T$ after the above modification,
%
the box $C$ 
is the root of
this tree.
Then, we count the number $k$ of messages to be distributed in $D$ multi-rounds,
by implementing the following simple recursive algorithm.
Let $k(C)$ be the number of messages from stations located originally in $C$.
First, each leaf $C'$ of $T$ sends $k(C')$ to $\pred(C')$. Each box $C'$ which is
not a leaf is waiting until receiving values $k(C'')$ from all elements of $\suc(C')$.
Then, it sets $k(C')\gets k(C')+\sum_{C''\in\suc(C')}k(C'')$ and sends $k(C')$ to
(the only element of) $\pred(C')$.

The actual multi-broadcast task is split into two stages. First, all messages are
collected in the root box $C$ of the tree $T$. In order to accomplish this task,
we apply a simple greedy algorithm for $D+k$ multi-rounds. Each box $C'$
stores the set $R(C')$ initially equal to the set of messages stored in stations
from $C'$ and the set $S(C')$ initially equal to the empty set.
The idea is that $R(C')$ stores messages received by (stations of) $C'$ from its
subtree of $T$ and not sent yet to $\pred(C')$, while $S(C')$ contains messages sent
already to $\pred(C')$.
In each multi-round, if $R(C')$ is not empty then $C'$ chooses arbitrary message $M$
from $R(C')$, removes $M$ from $R(C')$, adds $M$ to $S(C')$ and sends it to $\pred(C')$.
Observe that the task of collecting all messages in $C$ can be expressed
as an instance of the routing problem, where each message $M$ should be
routed from the box $C'$ containing the station in which $M$ is originally stored
to the root $C$ of $T$ along a (unique) path connecting $C'$ and $C$ in $T$.
Cidon et al.\ \cite{CidonKMP95} (Theorem~3.1) showed that such
a greedy
algorithm (that a vertex sends an arbitrary message each time it has message
not transmitted yet) finishes the routing task in $D+k-1$ time, provided
the input graph $G(V,E)$
is leveled (i.e., there exists a labeling $level: V\to [|V|]$ such
that for every directed edge $(u,v)$, $level(u)+1=level(v)$).
A tree is certainly a leveled graph, thus all messages are collected in $C$
after $k+D$ multi-rounds.

After collecting all $k$ messages,
$C$ starts the last stage of the algorithm which just the pipelined
flooding algorithm. That is, in each of the consecutive $k$ multi-rounds,
$C$ sends one of the messages to its neighbors in $T$.
Each box $C'$ in distance $j$ from $C$ receives these messages in rounds
$[j,j+k)$ and sends them to 
$\suc(C')$ in rounds~$[j+1,j+k+1)$.